\documentclass[12pt]{article}

\usepackage{amsmath}
\usepackage{amssymb}
\usepackage{amsfonts}
\usepackage{latexsym}
\usepackage{color}

\catcode `\@=11 \@addtoreset{equation}{section}

\catcode `\@=12

\newcommand{\be}{\begin{equation}}
\newcommand{\en}{\end{equation}}
\newcommand{\bea}{\begin{eqnarray}}
\newcommand{\ena}{\end{eqnarray}}
\newcommand{\beano}{\begin{eqnarray*}}
\newcommand{\enano}{\end{eqnarray*}}
\newcommand{\bee}{\begin{enumerate}}
\newcommand{\ene}{\end{enumerate}}

\newcommand{\J}{{\mathfrak J}}

\newcommand{\mc}{\mathcal}

\newcommand{\F}{{\cal F}}

\newcommand{\GG}{\mc G}

\newcommand{\1}{1 \!\! 1}

\newcommand{\Hil}{\mc H}

\newtheorem{thm}{Theorem}

\newtheorem{prop}[thm]{Proposition}
\newtheorem{defn}[thm]{Definition}

\newenvironment{proof}{\noindent {\bf Proof --}}{\hfill$\square$ \vspace{3mm}\endtrivlist}

\catcode `\@=11 \@addtoreset{equation}{section}
\catcode `\@=12

\begin{document}

\thispagestyle{empty}

\vspace*{2cm}

\begin{center}
{\Large \bf Non-isospectral Hamiltonians, intertwining operators and hidden hermiticity}   \vspace{2cm}\\

{\large F. Bagarello}\\
  Dipartimento di Metodi e Modelli Matematici,
Facolt\`a di Ingegneria,\\ Universit\`a di Palermo, I-90128  Palermo, Italy\\
e-mail: bagarell@unipa.it

\end{center}

\vspace*{2cm}

\begin{abstract}
\noindent We have recently proposed a  strategy to produce, starting from a given hamiltonian $h_1$ and a certain operator $x$ for which $[h_1,xx^\dagger]=0$ and $x^\dagger x$ is invertible, a second hamiltonian $h_2$ with the same eigenvalues as $h_1$ and whose eigenvectors are related to those of $h_1$ by $x^\dagger$. Here we extend this procedure to build up a second hamiltonian, whose eigenvalues are different from those of $h_1$, and whose eigenvectors are still related as before. This new procedure is also extended to crypto-hermitian hamiltonians.

\end{abstract}

\vspace{2cm}


\vfill


\newpage

\section{Introduction}

The problem of finding quantum system for which the Schr\"odinger equation can be solved exactly is not an easy task: obtaining the eigenvalues and the eigenvectors of a given hamiltonian is often possible only at a perturbative level. Hence, finding new potentials for which some non-perturbative solution can be found is a rather useful goal in theoretical physics, which has produced many results in the so-called super-symmetric quantum mechanics (SUSY qm), \cite{CKS}, and in the theory of intertwining operators (IO), see \cite{spi} and references therein.
Recently, \cite{bag1,baglast,bag4}, we have proposed a simple technique which produces, starting from a given hamiltonian $h_1$ with known eigensystem, a second hamiltonian $h_2$ with the same eigenvalues and eigenvectors which are easily deduced from those of $h_1$. In other words, this means that we produce, from a {\em solvable potential}, a second quantum potential which is also solvable. This is what, in the literature,  is called a  {\em Darboux transform}. Our technique relies, in particular, on the possibility of inverting  a certain operator $N_2$, see below. In this paper we show that, if $N_2$ is not invertible, the same general strategy could be used to construct, out of $h_1$, a second hamiltonian, which we again call $h_2$, whose eigenstates and eigenvalues are different from those of $h_1$ but which however can be easily computed from these ones. Moreover, within our new results, the requirement of $h_1$ to be self-adjoint is not really necessary and can be replaced by its crypto-hermiticity, in the sense of Znojil, \cite{znojil}: $h_1=\Theta^{-1}h_1^\dagger\Theta$, for a certain operator $\Theta$. In this case, we will show that  the  procedure discussed in Section II can be easily modified and several interesting results can still be obtained. In particular, we will see how to construct several sets of complete vectors in the Hilbert space $\Hil$, as well as several intertwining operators. It is maybe useful to stress here that crypto-hermitian quantum mechanics and its many {\em relatives}, see \cite{bagzno} and references therein, is nowadays a rather hot topic, both for its theoretical aspects and for some recent experiments which seem to fit well in this new scheme, \cite{mak}.

This paper is organized as follows: in the next section we discuss the general strategy and some of its consequences. In Section III we discuss two examples in which the eigenvalues of $h_1$ have multiplicity one. In Section IV we give examples with multiplicity larger than one. In Section V we extend our construction to crypto-hermitian hamiltonians. Section VI contains our conclusions.

\section{The strategy}

In a recent paper the possibility of constructing, from  a self-adjoint hamiltonian $h_1$, a second hamiltonian with the same eigenvalues and related eigenvectors have been discussed in detail,  \cite{bag1,baglast,bag4}. Let $h_1$ be a self-adjoint hamiltonian on the Hilbert space $\Hil_1$, $h_1=h_1^\dagger$, whose (not necessarily normalized) eigenvectors, $\varphi_n^{(1)}$, satisfy the following equation: $h_1\varphi_n^{(1)}=\epsilon_n\varphi_n^{(1)}$, $n\in\Bbb{N}_0:=\Bbb{N}\cup\{0\}$. Let $\Hil_2$ be a second Hilbert space, in general different from $\Hil_1$, and let us consider an operator $x:\Hil_2\rightarrow\Hil_1$, whose adjoint $x^\dagger$ maps $\Hil_1$ in $\Hil_2$. Let us further define $
N_1:=x\,x^\dagger$ and $N_2:=x^\dagger x$. These are surely well defined if $x$ is a bounded operator. On the other hand, if $x$ is unbounded, $N_2$ is well defined,   if, taken $f$ in the domain of $x$, $D(x)$, $xf\in D(x^\dagger)$. Analogously, $N_1$ is well defined if, taken $g$ in $D(x^\dagger)$, $x^\dagger g$ belongs to $D(x)$.
It is clear that $N_j$ maps $\Hil_j$ into itself, for $j=1,2$. Suppose now that $x$ is such that $N_2$ is invertible in $\Hil_2$ and that
$[N_1,h_1]=0$.
Of course, this commutator should be considered in a {\em weak form} if $h_1$ or $N_1$ is unbounded: $<N_1f,h_1g>_1=<h_1f,N_1g>_1$, for $f,g\in D(N_1)\cap D(h_1)$. Here $<,>_1$ is the scalar product in $\Hil_1$. Defining now
\be
\tilde h_2:=N_2^{-1}\left(x^\dagger\,h_1\,x\right),\qquad \varphi_n^{(2)}=x^\dagger\varphi_n^{(1)},
\label{21}\en
in \cite{bag1,baglast,bag4} it is shown that $\tilde h_2$ is self-adjoint,
 $\tilde h_2=\tilde h_2^\dagger$, that it satisfies the following modified version of intertwining relation
 $x^\dagger\left(x\,\tilde h_2-h_1\,x\right)=0$ and that,
 if $\varphi_n^{(2)}\neq 0$,  then $\tilde h_2\varphi_n^{(2)}=\epsilon_n\varphi_n^{(2)}$. Also, $[N_2,\tilde h_2]=0$, again in a weak form, in general.
Furthermore, if $\epsilon_n$ is non degenerate,  $\varphi_n^{(1)}$  and $\varphi_n^{(2)}$ are respectively eigenstates of $N_1$ and $N_2$  with the same eigenvalue.

In \cite{bag1,baglast,bag4} we have proposed several examples of this construction, some arising from the theory of the (g)-frames and some from quons, \cite{moh}. In particular, we have shown that it is convenient, if possible, to avoid any explicit representation of the operators involved in the definition of $h_1$ and $x$ and work, as much as possible, at a purely algebraic level, playing with the commutation relations. The main difficulty in the cited papers is the computation of $N_2^{-1}$, which however is  rarely needed (in the examples considered) since it usually disappears after some  re-ordering of the operators of $\tilde h_2$. But when this is not so, then some Green's function should be computed and this may not be a simple task, from a practical point of view. Therefore, the computation of $N_2^{-1}$ makes our strategy difficult to be applied. More than this, it could also happen that  $N_2$ has no inverse at all! Hence a natural question to answer is the following: if $N_2^{-1}$ does not exist, or if we are not able to compute it, are we still able to deduce a second {\em quantum system} whose eigenvalues and eigenvectors can be easily found? Luckily enough, the answer is affirmative, and this can be done with a simple extension of the  idea sketched above.

\vspace{3mm}

Let $h_1$ be, as before, a self-adjoint hamiltonian on the Hilbert space $\Hil_1$, $h_1=h_1^\dagger$, with normalized eigenvectors $\varphi_{n,k}^{(1)}$: $h_1\varphi_{n,k}^{(1)}=\epsilon_{n}^{(1)}\varphi_{n,k}^{(1)}$, $n\in\Bbb{N}_0$ and $k=1,2,\ldots m_n$, $m_n$ being the degeneracy of $\epsilon_{n}$. To simplify the notation, we will call $\J$ the set of these quantum numbers. Let $\Hil_2$ be a second Hilbert space, in general different from $\Hil_1$, and let  $x$ be an  operator from $\Hil_2$ to $\Hil_1$. In this paper we will mainly consider the case of $x$ bounded, but quite often  we will also comment on what happens for unbounded $x$. Let us further define $
N_1:=x\,x^\dagger$, $N_2:=x^\dagger x$. As already discussed, these operators are well defined if $x$ is bounded while some extra requirement has to be assumed for unbounded $x$. We assume here first that, for such $x$,
$[N_1,h_1]=0$ (in a weak sense, if needed). Nothing will be assumed on the existence of $N_2^{-1}$. Because of the commutativity between $h_1$ and $N_1$  $\F_1=\left\{\varphi_{n,k}^{(1)}, \,(n,k)\in J\right\}$ can be taken to be a family of eigenstates of $N_1$ as well, $N_1\varphi_{n,k}^{(1)}=\nu_{n,k}\varphi_{n,k}^{(1)}$, $\forall (n,k)\in \J$. We will assume here that, using the language of physicists, $h_1$ and $N_1$ are a {\em complete set of commuting observables}. In other words, the set $\F_1$ is a basis of $\Hil_1$, which is clearly  orthonormal (o.n.): $\left<\varphi_{n,k}^{(1)},\varphi_{m,l}^{(1)}\right>_1=\delta_{n,m}\delta_{k,l}$. This is not a big requirement since, if it is not true, we could always replace the {\em original} Hilbert space $\Hil_1$ with a {\em new} one, $\tilde\Hil_1$, constructed taking the closure of the linear span of $\F_1$. The closure relation of $\F_1$ in bra-ket language reads as follows:
\be
\sum_{(k,n)\in\J}\left|\varphi_{n,k}^{(1)}\left>\right<\varphi_{n,k}^{(1)}\right|=\1_1,
\label{23}\en
where $\1_1$ is the identity operator in $\Hil_1$. Due to the definition of $N_1$ it is clear that its eigenvalues $\nu_{n,k}$ cannot be negative. Indeed we find that $\nu_{n,k}=\|x^\dagger\varphi_{n,k}^{(1)}\|_2^2$, $\|.\|_2$ being the norm in $\Hil_2$, which is always positive and is zero if and only if $\varphi_{n,k}^{(1)}\in \ker(x^\dagger)$. Therefore, if  $\ker(x^\dagger)=\{0\}$, then $\nu_{n,k}>0$ for all $(n,k)\in\J$ and, as a consequence, $N_1$ admits inverse.

Let us now define
\be
 h_2:=x^\dagger\,h_1\,x,\qquad \varphi_{n,k}^{(2)}:=x^\dagger\varphi_{n,k}^{(1)}, \qquad \epsilon_{n,k}^{(2)}:=\epsilon_{n}^{(1)}\nu_{n,k}.
\label{24}\en
Notice that, in principle, $\varphi_n^{(2)}$ could be zero if $\ker(x^\dagger)\neq\{0\}$. To begin with, the following properties can be easily deduced:
\be
h_2=h_2^\dagger,\qquad [h_2,N_2]=0,\qquad N_1\,x=x\,N_2,\qquad \left(h_1N_1\right)\,x=x\,h_2.
\label{25}\en
With our definitions, therefore, $h_2$ is also self-adjoint and commutes with $N_2$ (weakly, if needed). From this point of view, $h_2$ and $N_2$ behave exactly as $h_1$ and $N_1$. Moreover, $x$ intertwines between $N_1$ and $N_2$, as well as between $h_1\,N_1$ and $h_2$, and this will have consequences on the related eigensystems. In particular we have that, if $\varphi_{n,k}^{(1)}\notin \ker(x^\dagger)$, then the non zero vector $\varphi_{n,k}^{(2)}$ satisfies the following eigenvalue equations
\be
h_2\varphi_{n,k}^{(2)}=\epsilon_{n,k}^{(2)}\varphi_{n,k}^{(2)},\qquad N_2\varphi_{n,k}^{(2)}=\nu_{n,k}\varphi_{n,k}^{(2)},
\label{26}\en
whose proofs are straightforward. Formula (\ref{26}) shows a first difference between $h_1$ and $h_2$: while the first has degenerate eigenvalues, in general, the eigenvalues of $h_2$ are not degenerate. In view of what has been discussed before we can also say that $\nu_{n,k}=0$ if and only if $\varphi_{n,k}^{(2)}=0$. Hence, if $\nu_{n,k}>0$, we can deduce that $\varphi_{n,k}^{(2)}\notin\ker(x)$ and, more than this, that
\be
\varphi_{n,k}^{(1)}=\frac{1}{\nu_{n,k}}\,x\,\varphi_{n,k}^{(2)},
\label{27}\en
which, in a sense, reverse the second equation in (\ref{24}).

Let now $\J'=\{(n,k)\in\J: \nu_{n,k}>0\}$. A consequence of the orthonormality of the set $\F_1$ is that also the functions of $\F_2=\left\{\varphi_{n,k}^{(2)}, \,(n,k)\in \J'\right\}$ are orthogonal but not normalized, in general. Indeed we have, taking $(n,k), (m,l)\in \J'$,
$$
\left<\varphi_{n,k}^{(2)},\varphi_{m,l}^{(2)}\right>_2=\left<x^\dagger\varphi_{n,k}^{(1)},
x^\dagger\varphi_{m,l}^{(1)}\right>_2=\left<N_1\varphi_{n,k}^{(1)},\varphi_{m,l}^{(1)}\right>_1=
\nu_{n,k}\delta_{n,m}\delta_{k,l}.
$$
Here, obviously, $<,>_2$ is the scalar product in $\Hil_2$. Let us now prove the following result:

\begin{prop}
 $\ker(x)=0$ if and only if $\F_2$ is complete in $\Hil_2$.
\end{prop}
\begin{proof}
We divide the proof of the statement in two parts: $\J=\J'$ and $\J'\subset\J$.

\vspace{2mm}{\bf First case: $\J=\J'$.}

As we have already discussed, since $\J=\J'$,   $\nu_{n,k}>0$ for all $(n,k)\in\J$.  Let us  prove first that, if $\ker(x)=0$ then $\F_2$ is complete in $\Hil_2$. For that we consider a vector $f\in D(x)\subseteq\Hil_2$ orthogonal to all the $\varphi_{n,k}^{(2)}$'s: $\left<f,\varphi_{n,k}^{(2)}\right>_2=0$, $\forall (n,k)\in\J$. Here  the domain of $x$, $D(x)$,  can be taken coincident with $\Hil_2$ itself if $x$ is bounded. Otherwise $D(x)$ is dense in $\Hil_2$; we want to deduce that $f=0$. First we notice that $xf=0$, since $\F_1$ is complete in $\Hil_1$ by assumption. But $\ker(x)=\{0\}$. Hence $f=0$. This ends the proof for $x$ bounded, while, if $x$ is unbounded, we simply use the density of $D(x)$ in $\Hil_2$.

\vspace{2mm}

Let us now suppose that $\F_2$ is complete. We will show that  $\ker(x)=\{0\}$. Let $f\in\ker(x)$. Hence $xf=0$.  Since $xf\in\Hil_1$ and since $\F_1$ is an o.n. basis of $\Hil_1$ we have
$$
xf=\sum_{(n,k)\in\J}\left<xf,\varphi_{n,k}^{(1)}\right>_1\varphi_{n,k}^{(1)}=
\sum_{(n,k)\in\J}\left<f,\varphi_{n,k}^{(2)}\right>_2\varphi_{n,k}^{(1)},
$$
where we have used the definition of $\varphi_{n,k}^{(2)}$ and the fact that $\J=\J'$. Then, taking the scalar product of the above expansion with $\varphi_{m,l}^{(1)}$ and recalling that $xf=0$, we deduce that, $\forall (n,k)\in\J$,
$$
0=\left<f,\varphi_{m,l}^{(2)}\right>_2,
$$
which implies that $f=0$ since $\F_2$ is complete by assumption. Hence $\ker(x)=\{0\}$.

\vspace{2mm}{\bf Second case: $\J'\subset\J$.}

To be  concrete, we will assume here that $\nu_{0,0}=0$, while all the others $\nu_{n,k}$ are strictly positive. Let us prove first that, also in this case, if $\ker(x)=\{0\}$, then $\F_2$ is complete. Let us assume that $\left<f,\varphi_{n,k}^{(2)}\right>_2=0$ for all $(n,k)\in\J'$. This implies that $xf=\alpha\varphi_{0,0}^{(1)}$, for some complex $\alpha$. Hence, $x^\dagger\,x\,f=\alpha\,x^\dagger\varphi_{0,0}^{(1)}=0$, since $\nu_{0,0}=0$, so that $\left<f,x^\dagger\,x\,f\right>_2=\|x\,f\|_2^2=0$. This implies that $f\in \ker(x)$, which only contains the zero vector by assumption. Hence $f=0$ and $\F_2$ is complete, as a consequence.

Let us now assume that $\F_2$ is complete. Then, if $f\in\ker(x)$, $xf=0$. Repeating the same steps  as before we have $0=xf=\sum_{(n,k)\in\J'}\left<f,\varphi_{n,k}^{(2)}\right>_2\varphi_{n,k}^{(1)}$, which obviously implies that $\left<f,\varphi_{m,l}^{(2)}\right>_2=0$ for all $\varphi_{m,l}^{(2)}\in\F_2$. But $\F_2$ is complete. Hence $f=0$ and $\ker(x)=\{0\}$.

\end{proof}

{\bf Remarks:--} (1) This Proposition is a somehow  {\em refined version} of a similar result contained in \cite{bag2010}. In particular we are here considering  the possibility of $x$ to be unbounded, and the possibility that not all the $\nu_{n,k}$ are strictly positive.

(2) There is an evident asymmetry between $\ker(x)$ and $\ker(x^\dagger)$ in this Proposition. The reason is clear: we are assuming that $\F_1$ is an o.n. basis of $\Hil_1$ and we have shown that, if $\ker(x)=\{0\}$, $\F_2$ is also a basis of $\Hil_2$. The role of $\ker(x)$ and $\ker(x^\dagger)$ would be exchanged if we start with the assumption that $\F_2$ is a basis and we ask for conditions which makes of the set $\F_1$ of vectors (\ref{27}) a basis of $\Hil_1$.

\section{Examples with multiplicity 1}

\subsection{standard bosons}

Let $a$ be the usual annihilation operator satisfying $[a,a^\dagger]=\1$, and let $\varphi_0$ be the vacuum of $a$: $a\varphi_0=0$. Then the set $\F:=\left\{\varphi_n:=\frac{1}{\sqrt{n!}}\,{a^\dagger}^n\varphi_0,\,n\geq0\right\}$ is an o.n. basis of $\Hil=\Hil_1=\Hil_2$. Let us further define $h_1=a^\dagger\,a=:\hat n$, the number operator. Hence $\varphi_n^{(1)}\equiv\varphi_n$ and $\epsilon_n^{(1)}=n$: $h_1\varphi_n^{(1)}=n\varphi_n^{(1)}$, $n=0,1,2,\ldots$. If we now take $x:=a^k$, for a fixed natural $k$, it is clear that $\ker(x)=\{\varphi_0, \varphi_1,\ldots,\varphi_{k-1}\}$, so that the set $\F_2$ will not be complete in $\Hil$. This can be seen explicitly since the vectors of $\F_2$ are the following: $\varphi_0^{(2)}=x^\dagger \varphi_0^{(1)}=\sqrt{k!}\,\varphi_k^{(1)}$, $\varphi_1^{(2)}=x^\dagger \varphi_1^{(1)}=\sqrt{(k+1)!}\,\varphi_{k+1}^{(1)}$, and so on. It is clear that, for instance, the non zero vector $\varphi_0^{(1)}$ is orthogonal to all the $\varphi_n^{(2)}$'s, so that $\F_2$ cannot be complete.
On the other hand, we can also check that $\ker(x^\dagger)=\{0\}$. The operator $N_1=xx^\dagger=a^k\,{a^\dagger}^k$ commutes with $h_1$ for all $k$, as can be checked using induction on $k$. Moreover, it is interesting to notice that $N_2={a^\dagger}^k{a}^k$ admits no inverse. Hence, the procedure proposed in \cite{bag1} does not apply here.

The hamiltonian $h_2=x^\dagger h_1 x={a^\dagger}^{k+1}\,a^{k+1}$ can be written in terms of the number operator $\hat n$ as follows:
\be
h_2=\left(\hat n-k\1\right)\left(\hat n-(k-1)\1\right)\cdots \left(\hat n-2\1\right)\left(\hat n-\1\right)\hat n,
\label{31}\en
whose proof can be again given by induction. The vectors $\varphi_n^{(2)}=x^\dagger \varphi_n^{(1)}=\sqrt{\frac{(k+n)!}{n!}}\,\varphi_{n+k}$ are eigenstates of $h_2$ with eigenvalues $\epsilon_n^{(2)}:=\frac{(n+k)!}{(n-1)!}$. Not surprisingly, $h_2$, $\varphi_n^{(2)}$ and $\epsilon_n$ all depend on $k$, which is a consequence of the fact that $x$ itself depends on $k$. We also find
$$
N_1\varphi_n^{(1)}=\nu_n\varphi_n^{(1)}, \qquad \nu_n=\frac{(n+k)!}{n!},
$$
so that, as expected, $\nu_n\epsilon_n^{(1)}=\epsilon_n^{(2)}$. Moreover $N_2=x^\dagger x$ coincides with $h_2$ with $k$ replaced by $k-1$. Hence it is clear that $[h_2,N_2]=0$. It is a simple exercise to check that all the properties listed in Section II are satisfied.

\subsection{generalizing this example}

Following \cite{bag4} we consider two operators, $B$ and $B^\dagger$, which satisfy the modified commutation relation $[B,B^\dagger]_q:=B\,B^\dagger-q B^\dagger B=\1$, $q\in[0,1]$. Let $\varphi_0^{(1)}$ be the vacuum of $B$: $B\varphi_0^{(1)}=0$. Let furthermore $h_1=B^\dagger B$. Then, putting
\be\varphi_n^{(1)}=\frac{1}{\beta_0\cdots\beta_{n-1}}\,{B^\dagger}^n\,\varphi_0^{(1)}=
\frac{1}{\beta_{n-1}}B^\dagger\varphi_{n-1}^{(1)},\qquad n\geq 1,\label{21ex}\en
we have $h_1\varphi_n^{(1)}=\epsilon_n^{(1)}\varphi_n^{(1)}$, with $\epsilon_0^{(1)}=0$, $\epsilon_1^{(1)}=1$ and $\epsilon_n^{(1)}=1+q+\cdots+q^{n-1}$ for $n\geq 1$. Also, the normalization is found to be $\beta_n^2=1+q+\cdots+q^n$, for all $n\geq0$. Hence $\epsilon_n^{(1)}=\beta_{n-1}^2$ for all $n\geq1$.  The set of the $\varphi_n^{(1)}$'s  spans the Hilbert space $\Hil=\Hil_1=\Hil_2$ and they are mutually o.n.: $<\varphi_{n}^{(1)},\varphi_{k}^{(1)}>=\delta_{n,k}$.

We now take, as for the ordinary bosons discussed before, $x={B}^k$. Again, its kernel is different from the zero vector. Hence $\F_2$ is not complete. The operator  $N_1=B^k{B^\dagger}^k$ commutes with $h_1$ for all fixed $k$. This can be seen, for instance, using the following recurrence relation: $N_1^{(k+1)}=N_1^{(k)}\left(q^k N_1^{(1)}+\epsilon_k^{(1)}\right)$, where we have introduced the suffix $k$ to make explicit the dependence on $k$ here. The eigenvalues of $N_1=N_1^{(k)}$ obey the following
recurrence rule: $\nu_n^{(1)}=1+q\epsilon_n^{(1)}$, $\nu_n^{(k+1)}=\nu_n^{(k)}(q^{k+1}\epsilon_n^{(1)}+\epsilon_{n+1}^{(1)})$. Notice that, in the limit $q\rightarrow 1^-$, we find $\epsilon_n^{(1)}\rightarrow n$, $\nu_n^{(1)}\rightarrow 1+n$.

With the usual definitions we find that $h_2=(B^\dagger)^{k+1}B^{k+1}$ and that $\varphi_n^{(2)}=x^\dagger\varphi_n^{(1)}$ coincides, but for a normalization, with $\varphi_{n+k}^{(1)}$ (which again shows that $\F_2$ is not complete). It is again a matter of simple but boring computations to check that all the properties of the previous section are satisfied.

\section{Examples from Landau levels}

Let us consider an  electron in a uniform magnetic field oriented in the positive $z$-direction, with vector potential $\vec A^\uparrow=\frac{B}{2}(-y,x,0)$. Its hamiltonian, $H^\uparrow=\frac{1}{2m}\left(\vec p-\frac{e}{c}\,\vec A^\uparrow\right)^2$, can be written as
\be
H^\uparrow=H_0+H_1^\uparrow=\hbar\omega\left(N_++N_-+\1\right)+\hbar\omega\left(N_--N_+\right)=\hbar\omega\left(2N_-+\1\right).
\label{41}\en
Here we have introduced $\omega=\frac{eB}{2mc}$, $a_x=\sqrt{\frac{m\omega}{2\hbar}}\,x+i\frac{1}{\sqrt{2m\omega\hbar}}\,p_x$, $a_y=\sqrt{\frac{m\omega}{2\hbar}}\,y+i\frac{1}{\sqrt{2m\omega\hbar}}\,p_y$, $A_\pm=\frac{1}{\sqrt{2}}\left(a_x\mp ia_y\right)$ and $N_\pm=A_\pm^\dagger A_\pm$. These operators satisfy the canonical commutation relations: $[a_x,a_x^\dagger]=[a_y,a_y^\dagger]=[A_\pm,A_\pm^\dagger]=\1$, all the other commutators being zero. Then, taking $\Psi_{0,0}$ such that $A_\pm\Psi_{0,0}=0$ and $\Psi_{n_+,n_-}:=\frac{1}{\sqrt{n_+!n_-!}}(A_+^\dagger)^{n_+}(A_-^\dagger)^{n_-}\Psi_{0,0}$, $n_\pm=0,1,2,\ldots$, we have
$$\left\{
\begin{array}{ll}
N_+\Psi_{n_+,n_-}=n_+\Psi_{n_+,n_-}; \quad N_-\Psi_{n_+,n_-}=n_-\Psi_{n_+,n_-}, \\
H_0\Psi_{n_+,n_-}=\hbar\omega\left(n_++n_-+1\right)\Psi_{n_+,n_-}, \\ H_1^\uparrow\Psi_{n_+,n_-}=\hbar\omega\left(n_--n_+\right)\Psi_{n_+,n_-},
\end{array}
\right.
$$
which implies that $H^\uparrow\Psi_{n_+,n_-}=\hbar\omega\left(2n_-+1\right)\Psi_{n_+,n_-}$.
If we rather start with a magnetic field in the negative $z$-direction, $\vec A^\downarrow=\frac{B}{2}(y,-x,0)$, since $H^\downarrow=H_0+H_1^\downarrow= H_0-H_1^\uparrow$, the eigenstates are again $\Psi_{n_+,n_-}$. In particular we find
$$\left\{
\begin{array}{ll}
H_1^\downarrow\Psi_{n_+,n_-}=\hbar\omega\left(n_+-n_-\right)\Psi_{n_+,n_-}, \\
H^\downarrow\Psi_{n_+,n_-}=\hbar\omega\left(2n_++1\right)\Psi_{n_+,n_-}. \\
\end{array}
\right.
$$
We conclude that the eigenvalues of both $H^\uparrow$ and $H^\downarrow$ are degenerate (each with infinite degeneracy), so that what discussed in Section II can be applied. Before going on we also need to define a map, considered for instance in \cite{abg}, which works like this: $j(\Psi_{n_+,n_-})=\Psi_{n_-,n_+}$, for all $n_+,n_-$. Hence it is easily seen that $j$ intertwines between $H^\uparrow$ and $H^\downarrow$: $j\,H^\uparrow=H^\downarrow \,j$.

\subsection{A first choice of $x$}

We take here $h_1:=H^\uparrow=\hbar\omega\left(2N_-+\1\right)$ and $x=A_+A_-$. With this choice $N_1=xx^\dagger=(N_++\1)(N_-+\1)$ and $[h_1,N_1]=0$. It is clear that $\varphi_{\bf n}^{(1)}:=\Psi_{\bf n}$, ${\bf n}=(n_+,n_-)$, is an eigenstate of $h_1$ and $N_1$, with eigenvalues $\epsilon_{n_-}^{(1)}:=\hbar\omega\left(2n_-+1\right)$ and $\nu_{\bf n}:=(n_++1)(n_-+1)$, respectively. It is clear that $\nu_{\bf n}>0$ and that $x^\dagger\varphi_{\bf n}^{(1)}\neq 0$, for all $\bf n$.

Defining now $h_2=x^\dagger h_1 x$, $N_2=x^\dagger x$, and $\varphi_{\bf n}^{(2)}=x^\dagger\varphi_{\bf n}^{(1)}$, and playing a little bit with the commutation relations, we deduce that $$h_2=\hbar\omega\,N_+ N_-\left(2N_--\1\right),\quad \varphi_{\bf n}^{(2)}=\sqrt{(n_++1)(n_-+1)}\,\Psi_{n_-+1,n_++1},\quad N_2=N_+N_-,$$ (which is not invertible, by the way!). Defining $\epsilon_{\bf n}^{(2)}=\epsilon_{n}^{(1)}\nu_{\bf n}=\hbar\omega(2n_-+1)(n_++1)(n_-+1)$, it is now easy to check that all our claims are satisfied. For instance $h_2=h_2^\dagger$, $[h_2,N_2]=0$, $N_1x=xN_2$, $(h_1N_1)x=xh_2$. The vectors $\varphi_{\bf n}^{(2)}$ satisfy $h_2\varphi_{\bf n}^{(2)}=\epsilon_{\bf n}^{(2)}\varphi_{\bf n}^{(2)}$ and $N_2\varphi_{\bf n}^{(2)}=\nu_{\bf n}^{(2)}\varphi_{\bf n}^{(2)}$. They are orthogonal but not normalized, in general: $\left<\varphi_{\bf n}^{(2)},\varphi_{\bf m}^{(2)}\right>=\nu_{\bf n}^{(2)}\delta_{{\bf n},{\bf m}}$. Moreover, since $\ker(x)=\{0,\Psi_{0,n_-},\Psi_{n_+,0};\,n_-,n_+\geq0\}$, which is infinite dimensional, the set $\F_2$ is not expected to be complete in $\Hil$. Indeed, this is so since we can check that, for instance, the non-zero vector $\Psi_{0,0}$ is orthogonal to the vectors in $\F_2$.

\subsection{a second choice of $x$: mixing the quantum numbers}

As before we take $h_1:=H^\uparrow=\hbar\omega\left(2N_-+\1\right)$. Now, to discuss the effect of the map $j$, we consider $x=A_+\,j$. Hence $N_1=N_++\1$, and $[h_1,N_1]=0$. Further we find $h_2=x^\dagger h_1 x=\hbar\omega N_-\left(2N_++\1\right)$, $N_2=N_-$ and $\varphi_{\bf n}^{(2)}=x^\dagger\varphi_{\bf n}^{(1)}=\sqrt{n_++1}\,\varphi_{n_-,n_++1}^{(1)}$. Notice that $N_1$ is degenerate as well, but the eigenvalues of $(h_1,N_1)$ together uniquely fix the eigenvector of the system. As in the previous example, all the properties stated in Section II are recovered explicitly. Moreover, since all the vectors $\varphi_{n_+,0}^{(1)}$ belong to the kernel of $X$, $\F_2$ is not complete. This is related to the fact that $\ker(x)\neq\{0\}$.

Notice that the quantum numbers $(n_+,n_-)$ in  $\varphi_{\bf n}^{(2)}$ appear to be exchanged with respect to those of $\varphi_{\bf n}^{(1)}$, and the second number is also changed by one unit. This is, in part, the effect of the map $j$ and, as discussed in \cite{abg}, can be related to the appearance of
 analytic and anti-analytic Hermite polynomials in the analysis of Landau levels.

\section{Crypto-hermiticity}

In this section we will show how losing the self-adjointness of $h_1$, rather than being a problem, gives rise to a lot of extra features enriching the structure, at least under suitable conditions. For this reason, we first recall what is meant by {\em crypto-hermiticity} of an operator, using the definition given in \cite{bagzno}:

\begin{defn}
Let us consider two operators $H$ and
$\Theta$ acting on the Hilbert space $\Hil$, with $\Theta$  positive
and invertible. Let  $H^\dagger$ be the adjoint of $H$ in $\Hil$
with respect to its scalar product and let
$H^\ddagger=\Theta^{-1}H^\dagger\Theta$, when this exists. We will
say that $H$ is crypto-hermitian with respect to $\Theta$
(CHwrt$\Theta$) if $H=H^\ddagger$.
\end{defn}

Notice that this definition reduces to the standard self-adjointness of $H$ if $\Theta=\1$. Using standard facts on functional calculus, the assumptions on $\Theta$ imply  that the operators $\Theta^{\pm 1/2}$ are well defined. Hence we can introduce another operator $h:=\Theta^{1/2}\,H\,\Theta^{-1/2}$, at least if the domains of the operators allow us to do so. More explicitly, $h$ is well defined if, taken $f\in D(\Theta^{-1/2})$,
$\Theta^{-1/2}f\in D(H)$ and if $H\,\Theta^{-1/2}f\in D(\Theta^{1/2})$. Of course,
these requirements are surely satisfied if $H$ and $\Theta^{\pm 1/2}$ are bounded. Otherwise some care is required. It is easy to check that $h=h^\dagger$.
\vspace{3mm}

The starting point of our analysis is now an operator $H_1$ which is not self-adjoint but which is CHwrt$\Theta$, $\Theta$ as above. Then $H^\ddagger=\Theta^{-1}H^\dagger\Theta=H$. Also, we assume that an operator $X$ exists such that, calling $N_1=XX^\ddagger$ and $N_2=X^\ddagger\,X$, we have, first of all, $[H_1,N_1]=0$. We notice that, being $\ddagger$ an adjoint map, $N_j=N_j^\ddagger$, $j=1,2$. In other words, $H_1$, $N_1$ and $N_2$ are all CHwrt$\Theta$. To simplify the analysis we will work in a single Hilbert space $\Hil$.
 All throughout  this section we will assume that the operator $h_1:=\Theta^{1/2}\,H_1\,\Theta^{-1/2}$ is well defined. In particular, this is so when $H_1$, $\Theta$ and $\Theta^{-1}$ are bounded. Then $h_1$ is self-adjoint, $h_1=h_1^\dagger$, and commutes with $\hat n_1:=\Theta^{1/2}\,N_1\,\Theta^{-1/2}$ which is also self-adjoint $\hat n_1=\hat n_1^\dagger$. Hence we have two commuting, self-adjoint, operators which can be simultaneously diagonalized. Therefore, there exists a family of vectors $\F^{(1)}_\varphi=\{\varphi_{n,k}^{(1)},\, (n,k)\in\J\}$, such that
\be
\left\{
\begin{array}{ll}
h_1\varphi_{n,k}^{(1)}=\epsilon_n^{(1)}\varphi_{n,k}^{(1)},\\
\hat n_1\varphi_{n,k}^{(1)}=\nu_{n,k}\varphi_{n,k}^{(1)},
\end{array}
\right.
\label{51}\en
for all $(n,k)\in\J$. We see that we are thinking of a possible degeneracy of $h_1$, degeneracy which is lifted by $\hat n_1$. We will assume that $\F^{(1)}_\varphi$ is an o.n. basis of $\Hil$ and that $\Theta^{\pm1/2}$ are bounded. Now, due to our assumptions on $\Theta$, it is clear that $\ker(\Theta^{\pm1})=\ker(\Theta^{\pm1/2})=\{0\}$. Hence, calling $\Phi_{n,k}^{(1)}=\Theta^{-1/2}\varphi_{n,k}^{(1)}$, the set $\F^{(1)}_\Phi=\{\Phi_{n,k}^{(1)},\, (n,k)\in\J\}$ is a Riesz basis of $\Hil$. It is also clear that they are eigenstates of $H_1$ and $N_1$:
\be
\left\{
\begin{array}{ll}
H_1\Phi_{n,k}^{(1)}=\epsilon_n^{(1)}\Phi_{n,k}^{(1)},\\
N_1\Phi_{n,k}^{(1)}=\nu_{n,k}\Phi_{n,k}^{(1)}.
\end{array}
\right.
\label{52}\en
The frame operator associated to  $\F^{(1)}_\Phi$ can be easily computed using the resolution of the identity for  $\F^{(1)}_\varphi$: $S_\Phi^{(1)}=\sum_{(k,n)\in\J}\left|\Phi_{n,k}^{(1)}\left>\right<\Phi_{n,k}^{(1)}\right|=\Theta^{-1}$. It is now very easy to construct a second Riesz basis, $\F^{(1)}_\eta=\{\eta_{n,k}^{(1)},\, (n,k)\in\J\}$, which is biorthogonal to $\F^{(1)}_\Phi$. Its vectors are defined as $\eta_{n,k}^{(1)}=\Theta^{1/2}\varphi_{n,k}^{(1)}=\Theta\Phi_{n,k}^{(1)}$, and we get, as expected, $S_\eta^{(1)}=\sum_{(k,n)\in\J}\left|\eta_{n,k}^{(1)}\left>\right<\eta_{n,k}^{(1)}\right|=\Theta={S_\Phi^{(1)}}^{-1}$. It is trivial to check that $\sum_{(k,n)\in\J}\left|\Phi_{n,k}^{(1)}\left>\right<\eta_{n,k}^{(1)}\right|=\1$ and that $\left<\Phi_{n,k}^{(1)},\eta_{m,l}^{(1)}\right>=\delta_{n,m}\delta_{k,l}$, as well as
\be
\left\{
\begin{array}{ll}
H_1^\dagger\eta_{n,k}^{(1)}=\epsilon_n^{(1)}\eta_{n,k}^{(1)},\\
N_1^\dagger\eta_{n,k}^{(1)}=\nu_{n,k}\eta_{n,k}^{(1)}.
\end{array}
\right.
\label{53}\en
These results reflect, essentially, those found in \cite{bagzno}. Here, however, these results are, in a certain sense, {\em doubled}. Indeed, extending what we have done in Section II, let us now define a new operator,
$H_2:=X^\ddagger H_1X$, and the new vectors $\Phi_{n,k}^{(2)}=X^\ddagger\Phi_{n,k}^{(1)}$, $(k,n)\in\J$. It is possible to extend to the present context properties analogous to those in (\ref{25}). In particular we find that
\be
H_2=H_2^\ddagger,\qquad [H_2,N_2]=0,\qquad N_1X=XN_2,\qquad H_1N_1X=XH_2.
\label{54}\en
Moreover, extending again the results of Section II, the set $\F^{(2)}_\Phi=\{\Phi_{n,k}^{(2)},\, (n,k)\in\J\}$ is complete in $\Hil$ if and only if $\ker({X^\ddagger}^\dagger)=\{0\}$, or, equivalently, if $\ker(X\Theta^{-1})=\{0\}$. Under this requirement, recalling that the different $\Phi_{n,k}^{(2)}$ are also linearly independent, it follows that $\F^{(2)}_\Phi$ is a basis of $\Hil$, whose frame operator can be written as follows:
\be
S_\Phi^{(2)}=\sum_{(k,n)\in\J}\left|\Phi_{n,k}^{(2)}\left>\right<\Phi_{n,k}^{(2)}\right|=X^\ddagger\Theta^{-1}{X^\ddagger}^\dagger=\Theta^{-1/2}X_\Theta^\dagger X_\Theta\Theta^{-1/2},
\label{55}\en
where $X_\Theta=\Theta^{1/2}X\Theta^{1/2}$. It is possible to check that $S_\Phi^{(2)}$ admits inverse. This is easily seen if, for simplicity, $D({X^\ddagger}^\dagger)=\Hil$ and if $\ker(X)=\{0\}$: in this case, for each $f\in\Hil$, $\left<f,S_\Phi^{(2)}f\right>=\left<g,\Theta^{-1}g\right>$, with $g:={X^\ddagger}^\dagger f$. Since $g\neq0$ and since $\ker(\Theta)=\{0\}$, $\left<f,S_\Phi^{(2)}f\right>>0$. Hence $S_\Phi^{(2)}$ can be inverted and, as a consequence of equation (\ref{55}),  also $(X_\Theta^\dagger X_\Theta)^{-1}$ exists in $\Hil$. Moreover, calling $\epsilon_{n,k}^{(2)}=\epsilon_{n}^{(1)}\nu_{n,k}$
\be
\left\{
\begin{array}{ll}
H_2\Phi_{n,k}^{(2)}=\epsilon_{n,k}^{(2)}\Phi_{n,k}^{(2)},\\
N_2\Phi_{n,k}^{(2)}=\nu_{n,k}\Phi_{n,k}^{(2)}.
\end{array}
\right.
\label{56}\en

Repeating here what we did above for the {\em first family of hamiltonians} $\GG_1:=(H_1,H_1^\dagger,h_1)$, we put  $h_2:={S_\Phi^{(2)}}^{-1/2}\,H_2\,{S_\Phi^{(2)}}^{1/2}$, $\hat n_2:={S_\Phi^{(2)}}^{-1/2}\,N_2\,{S_\Phi^{(2)}}^{1/2}$ and $\varphi_{n,k}^{(2)}={S_\Phi^{(2)}}^{-1/2}\varphi_{n,k}^{(1)}$, $(n,k)\in\J$\footnote{Notice that in the definitions of $\GG_1$ we used $\Theta^{-1}$ rather than $S_\Phi^{(1)}$ since they coincide.}.

The same statements concerning $\GG_1$ can now be extended to $\GG_2:=(H_2,H_2^\dagger,h_2)$, at least if $\ker(X\Theta^{-1})=\{0\}$. In this case, among the other properties, we can prove that $\F^{(2)}_\varphi=\{\varphi_{n,k}^{(2)},\, (n,k)\in\J\}$ is an o.n. basis of $\Hil$ and that $h_2=h_2^\dagger$. We can also construct the biorthogonal set $\F^{(2)}_\eta=\{\eta_{n,k}^{(2)},\, (n,k)\in\J\}$, with $\eta_{n,k}^{(2)}={S_\Phi^{(2)}}^{-1}\Phi_{n,k}^{(2)}={S_\Phi^{(2)}}^{-1/2}\varphi_{n,k}^{(2)}$, which are eigenstates of $H_2^\dagger$ and $N_2^\dagger$:
\be
\left\{
\begin{array}{ll}
H_2^\dagger\eta_{n,k}^{(2)}=\epsilon_{n,k}^{(2)}\eta_{n,k}^{21)},\\
N_2^\dagger\eta_{n,k}^{(2)}=\nu_{n,k}\eta_{n,k}^{(2)}.
\end{array}
\right.
\label{57}\en
In analogy with what we have done before, we can further define $S^{(2)}_\eta$ and we get $S^{(2)}_\eta={S_\Phi^{(2)}}^{-1}$. The only difference is that we don't know if $\F^{(2)}_\eta$ and $\F^{(2)}_\Phi$ are Riesz bases or not, since this is related to the boundedness of the operators $S^{(2)}_\eta$ and $S^{(2)}_\Phi$. As for the intertwining equations, our construction gives rise to many of them. We just list here the following:
$$
S_\Phi^{(1)}\,H_1^\dagger=H_1{S_\Phi^{(1)}},\qquad \mbox{ and }\qquad S_\Phi^{(2)}\,H_2^\dagger=H_2{S_\Phi^{(2)}}.
$$
Other relations involving $h_1$, $h_2$,  $S_\eta^{(1)}$ and $S_\eta^{(2)}$ can be easily deduced.

This section show how a rather rich framework can be constructed by just three main ingredients: an operator $\Theta$ positive and possibly bounded with bounded inverse, a second operator $H_1$ which is CHwrt$\Theta$, and, last but not least, a third operator, $X$, such that $[H_1,XX^\ddagger]=0$. While $\Theta$ and $H_1$ are all is needed in the construction of $\GG_1$, $X$ is the main ingredient to move to $\GG_2$, doubling the results originally deduced for $\GG_1$. We should stress that two interesting features {\em break the symmetry} between $\GG_1$ and $\GG_2$: the first one is that, while $H_1$ is degenerate, $H_2$ is not. This is because its eigenvalues depend on both $n$ and $k$. The second is a bit more subtle:  in the first part of the construction we move from $H_1$ and $N_1$ to the commuting self-adjoint operators $h_1$ and $\hat n_1$. Since they can be simultaneously diagonalized, we use the set $\F_\varphi^{(1)}$ of their eigenvectors to construct $\F_\Phi^{(1)}$ and, from this, $\F_\eta^{(1)}$, both being (Riesz) bases of $\Hil$. On the other hand, we give conditions for $\F_\Phi^{(2)}$ to be a basis of $\Hil$. Then this is automatically a set of eigenstates of $H_2$ and $N_2$, which are used to construct the rest of the structure, and in particular $\F_\varphi^{(2)}$ and $\F_\eta^{(2)}$. For instance, as already stated, $\F_\Phi^{(2)}$ and $\F_\eta^{(2)}$ are not guaranteed to be Riesz bases.

\vspace{3mm}

{\bf Remark:--} We end this section recalling that in \cite{bagzno} we have discussed the strong relations between crypto-hermitian operators and non-linear regular pseudo-bosons (NLRPB). This has an immediate consequence here: all the results discussed here could be restated for NLRPB as well. This aspect will not be considered here.

\section{Conclusions}

In this paper we have discussed a procedure to construct, starting from a self adjoint operator $h_1$ and a second operator $x$ such that $[h_1,xx^\dagger]=0$, another operator whose eigenvectors can be deduced from those of $h_1$. Some examples arising from bosons, quons and Landau levels have been discussed.

In the second part of the paper we have extended this construction to crypto-hermitian hamiltonians, showing that the procedure, in this case, can be still settled up and that the results are {\em doubled}.

\section*{Acknowledgements}

  The author acknowledges financial support by the Murst.

\end{document}